\documentclass[prr,twocolumn,aps,amsmath,amssymb,superscriptaddress,floatfix,longbibliography]{revtex4-1}

\usepackage{latexsym}
\usepackage{amssymb}
\usepackage{graphicx}
\usepackage{amsmath}
\usepackage{bm}
\usepackage{verbatim}
\usepackage{mathrsfs}
\usepackage{extarrows}
\usepackage{comment}
\usepackage{mathtools,slashed}
\usepackage{soul}
\usepackage[toc,page]{appendix}
\usepackage[vcentermath]{youngtab}
\usepackage{multirow}
\usepackage{atbegshi,picture}
\usepackage{lipsum}
\usepackage{amsthm}
\usepackage[colorlinks=true,linktoc=page,linkcolor=blue,citecolor=blue]{hyperref}

\newtheorem{theorem}{Theorem}
\newtheorem{lemma}[theorem]{Lemma}

\def\infinity{\infty}

\def\:={\,\raisebox{0.85pt}{.}\hspace{-2.78pt}\raisebox{2.85pt}{.}\!\!=\,}
\def\=:{\,=\!\!\raisebox{0.85pt}{.}\hspace{-2.78pt}\raisebox{2.85pt}{.}\,}

\begin{document}

\title{Bounds on quantum adiabaticity in driven many-body systems\\
from generalized orthogonality catastrophe and quantum speed limit}

\author{Jyong-Hao Chen}
\email{jhchen@lorentz.leidenuniv.nl}
\affiliation{Instituut-Lorentz, Universiteit Leiden, P.O.\ Box 9506, 2300 RA Leiden, The Netherlands}
\author{Vadim Cheianov}
\affiliation{Instituut-Lorentz, Universiteit Leiden, P.O.\ Box 9506, 2300 RA Leiden, The Netherlands}

\date{\today}    

\begin{abstract}
We provide two inequalities for estimating adiabatic fidelity in terms of two other more handily calculated quantities, 
i.e., generalized orthogonality catastrophe and quantum speed limit.
As a result of considering a two-dimensional subspace spanned by the initial ground state and its orthogonal complement, 
our method leads to stronger bounds on adiabatic fidelity than those previously obtained.
One of the two inequalities is nearly sharp when the system size is large, as illustrated using a driven Rice-Mele model, 
which represents a broad class of quantum many-body systems whose overlap of different instantaneous ground states exhibits orthogonality catastrophe.
\end{abstract}

\maketitle

{\bf Introduction.---}
The celebrated quantum adiabatic theorem (QAT) is a fundamental theorem in quantum mechanics \cite{Born27,Born28,Kato50,Messiah14},
which has many applications ranging from the Gell-Mann and Low formula in quantum field theory \cite{Nenciu89,Fetter03},
Born-Oppenheimer approximation in atomic physics \cite{Born27b,Ziman01,Car85}, and adiabatic transport in solid-state physics \cite{Thouless83,Berry84,Avron88,Avron95}
to adiabatic quantum computation \cite{Farhi00,Roland02,Albash18} and adiabatic quantum state manipulation \cite{Ivanov01,Budich13} in quantum technology.
In its simplest form, the QAT states that if the initial state of a quantum system is one of the eigenstates of a time-dependent Hamiltonian, which describes the system, 
and if the time variation of the Hamiltonian is {\it slow enough}, 
then the state of the system at a later time will still be close to the instantaneous eigenstate of the Hamiltonian.
To be more specific, 
we are interested in the time-dependent Hamiltonian $H^{\,}_{\lambda}$ whose dependence on time $t$ is through an implicit function $\lambda(t).$
For each $\lambda,$ the instantaneous ground state $|\Phi^{\,}_{\lambda}\rangle$ obeys the instantaneous eigenvalue equation 
of $H^{\,}_{\lambda}$,
\begin{align}
H^{\,}_{\lambda}|\Phi^{\,}_{\lambda}\rangle=E^{(0)}_{\lambda}|\Phi^{\,}_{\lambda}\rangle,
\label{eq: instantaneous eigenvalue equation}
\end{align}
with $E^{(0)}_{\lambda}$ being the instantaneous ground state energy,
whereas the physical state $|\Psi^{\,}_{\lambda}\rangle$ is the solution to the scaled time-dependent Schr\"odinger equation ($\hbar\equiv1$),
\begin{align}
\mathrm{i}\Gamma\frac{\partial}{\partial \lambda}|\Psi^{\,}_{\lambda}\rangle=H^{\,}_{\lambda}|\Psi^{\,}_{\lambda}\rangle,
\label{eq: time-dependent SE}
\end{align}
with $|\Psi^{\,}_{0}\rangle=|\Phi^{\,}_{0}\rangle$ by preparing the initial state, $|\Psi^{\,}_{0}\rangle$, to be the same as the ground state of $H^{\,}_{\lambda=0}$,
$|\Phi^{\,}_{0}\rangle.$ 
Here, $\Gamma\:=\partial^{\,}_{t}\lambda$ is the driving rate.

Mathematically, the QAT states that for however small $\epsilon>0$
and arbitrary value of $\lambda$, there exists a driving rate $\Gamma$ small enough such that
\begin{align}
1-\mathcal{F}(\lambda)\,<\,\epsilon,
\label{thm: QAT}
\end{align} 
where the {\it fidelity of adiabatic evolution} (for short, {\it adiabatic fidelity}),   
\begin{align}
\mathcal{F}(\lambda)=|\langle\Phi^{\,}_{\lambda}|\Psi^{\,}_{\lambda}\rangle|^2,
\label{eq: adiabatic fidelity}
\end{align} 
is the fidelity between the physical state
$|\Psi^{\,}_{\lambda}\rangle$
and the instantaneous ground state $|\Phi^{\,}_{\lambda}\rangle$.
The QAT is a powerful asymptotic statement. 
However, in certain contexts, it is not sufficient because one would like to know how quickly $\mathcal{F}(\lambda)$ 
approaches unity with decreasing the driving rate.
This is generally a hard problem since it is difficult to
compute the adiabatic fidelity (\ref{eq: adiabatic fidelity}) for generic quantum many-body systems
by directly solving the instantaneous eigenvalue equation (\ref{eq: instantaneous eigenvalue equation}) 
and the time-dependent Schr\"odinger equation (\ref{eq: time-dependent SE}).
Moreover, most of the existing literature merely proves the existence of the QAT in various settings \cite{Kato50,Avron87,Avron99,Jansen07,Bachmann17}
but rarely provides useful and practical tools for computing the adiabatic fidelity (\ref{eq: adiabatic fidelity}) quantitatively.

To make further progress, an insight proposed in Ref.\ \cite{Lychkovskiy17} is to compare the physical state $|\Psi^{\,}_{\lambda}\rangle$ 
and the instantaneous ground state $|\Phi^{\,}_{\lambda}\rangle$ for a given $\lambda$
with their common initial state $|\Psi^{\,}_{0}\rangle=|\Phi^{\,}_{0}\rangle.$ 
We now introduce these two extra ingredients in turn.
First, the fidelity between the instantaneous ground state $|\Phi^{\,}_{\lambda}\rangle$ and its initial state $|\Phi^{\,}_{0}\rangle$,
\begin{align}
\mathcal{C}(\lambda)\:=|\langle\Phi^{\,}_{\lambda}|\Phi^{\,}_{0}\rangle|^2,
\label{eq: orthogonality catastrophe}
\end{align}
is referred to as {\it generalized orthogonality catastrophe}. 
This name is motivated by Anderson's orthogonality catastrophe \cite{Anderson67,Gebert14},
which states that the overlap between ground states in Fermi gases with and without local scattering potentials vanishes as the system size approaches infinity.
In a later section, we are interested in a wide range of classes of time-dependent many-body Hamiltonians
whose generalized orthogonality catastrophe $\mathcal{C}(\lambda)$ 
decays exponentially with the system size and $\lambda^{2}$.
Second, the fidelity between the physical state $|\Psi^{\,}_{\lambda}\rangle$ and its initial state $|\Psi^{\,}_{0}\rangle$,
\begin{align}
F(\Psi^{\,}_{\lambda},\Psi^{\,}_{0})\:=|\langle\Psi^{\,}_{\lambda}|\Psi^{\,}_{0}\rangle|^2,
\label{eq: define fidelity}
\end{align}
is another useful quantity since the corresponding {\it Bures angle} $D(\Psi^{\,}_{\lambda},\Psi^{\,}_{0})$ 
\cite{Nielsen10,Bengtsson17},
\begin{align}
D(\Psi^{\,}_{\lambda},\Psi^{\,}_{0})\:= \frac{2}{\pi}\arccos\sqrt{F(\Psi^{\,}_{\lambda},\Psi^{\,}_{0})},
\label{eq: define Bures angle}
\end{align}
is upper bounded by using a version of the quantum speed limit \cite{Pfeifer93a,Pfeifer95}
\footnote{
Also, refer to Supplemental Material {\bf S1} for a simple alternative derivation
of the quantum speed limit (\ref{eq: bound from QSL})
},
\begin{subequations}
\label{eq: bound from QSL}
\begin{align}
\frac{\pi}{2}
D(\Psi^{\,}_{\lambda},\Psi^{\,}_{0})\leq
\min\left(\mathcal{R}(\lambda),\frac{\pi}{2}\right)
\equiv
\widetilde{\mathcal{R}}(\lambda),
\label{eq: bound from QSL a}
\end{align}
where
\begin{align}
&\mathcal{R}(\lambda)\:=
\int^{\lambda}_{0}
\frac{\mathrm{d}\lambda'}{|\partial^{\,}_{t}\lambda'|}
\Delta E^{\,}_{0}(\lambda'),
\label{eq: bound from QSL b}
\\
&\Delta E^{\,}_{0}(\lambda')\:=
\sqrt{
\langle\Psi^{\,}_{0}|H^{2}_{\lambda'}|\Psi^{\,}_{0}\rangle
-
\langle\Psi^{\,}_{0}|H^{\,}_{\lambda'}|\Psi^{\,}_{0}\rangle^{2}
}.
\label{eq: bound from QSL c}
\end{align}
\end{subequations}
Note that in Eq.\ (\ref{eq: bound from QSL a}) we have taken into account the fact that the Bures angle $D\left(\Psi^{\,}_{\lambda},\Psi^{\,}_{0}\right)\in[0,1]$ 
by its definition
as the function $\mathcal{R}(\lambda)$ defined in Eq.\ (\ref{eq: bound from QSL b}) is not guaranteed to be upper bounded by $\pi/2.$
The quantum speed limit is essentially a measure of how fast a quantum system can evolve.
Since the Bures angle $D(\Psi^{\,}_{\lambda},\Psi^{\,}_{0})$ measures a distance between two states,
the quantum uncertainty $\Delta E^{\,}_{0}(\lambda')$ 
in Eq.\ (\ref{eq: bound from QSL}) plays the role of speed.
Although there is not one single quantum speed limit,
we work with the version shown in Eq.\ (\ref{eq: bound from QSL b}) that enables its computation by knowing merely the Hamiltonian
and the initial state.
It is worth mentioning that a recent  theoretical study \cite{Campo21} suggests that quantum speed limits can be probed in cold-atom experiments.

Utilizing the generalized orthogonality catastrophe (\ref{eq: orthogonality catastrophe}) and the quantum speed limit (\ref{eq: bound from QSL}) as additional ingredients,
the main result of Ref.\ \cite{Lychkovskiy17}
is an inequality providing an upper bound for the difference 
between the adiabatic fidelity $\mathcal{F}(\lambda)$ (\ref{eq: adiabatic fidelity}) 
and the generalized orthogonality catastrophe $\mathcal{C}(\lambda)$ (\ref{eq: orthogonality catastrophe}),
\begin{align}
\left|
\mathcal{F}(\lambda)-\mathcal{C}(\lambda)
\right|
&\;\leq\; 
\widetilde{\mathcal{R}}(\lambda),
\label{eq: previous result}
\end{align}
where 
$\widetilde{\mathcal{R}}(\lambda)$ is defined in Eq.\ (\ref{eq: bound from QSL a}).
In this work, we derive two improved inequalities that are stronger than the inequality (\ref{eq: previous result}).
A schematic summary of our main results is illustrated in Fig.\ \ref{Fig: outline}.

\begin{figure}[t]
\begin{center}
\includegraphics[width=0.35\textwidth]{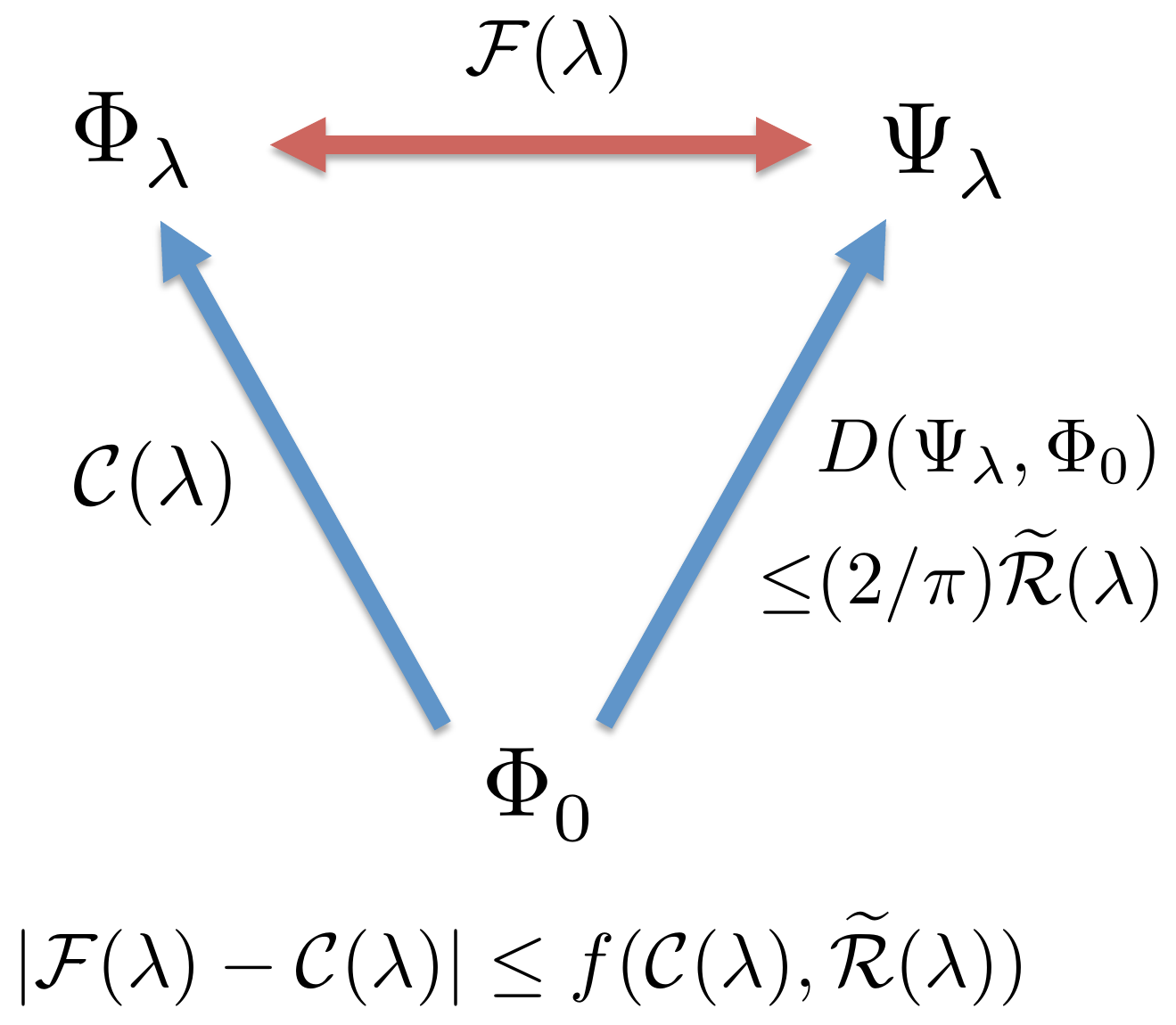}
\caption{
(Color online)
Schematic illustration of the relation between the physical state $|\Psi^{\,}_{\lambda}\rangle$, 
the instantaneous ground state $|\Phi^{\,}_{\lambda}\rangle,$ 
and the initial state $|\Psi^{\,}_{0}\rangle=|\Phi^{\,}_{0}\rangle.$  
The central object is the fidelity between $|\Psi^{\,}_{\lambda}\rangle$ and $|\Phi^{\,}_{\lambda}\rangle$, i.e., $\mathcal{F}(\lambda)$ (\ref{eq: adiabatic fidelity}),
which can be estimated through 
$\left|
\mathcal{F}(\lambda)-\mathcal{C}(\lambda)
\right|
\leq
f(\mathcal{C}(\lambda),\widetilde{\mathcal{R}}(\lambda)),$
where $f(\mathcal{C}(\lambda),\widetilde{\mathcal{R}}(\lambda))\in$
$\{\widetilde{\mathcal{R}}(\lambda),\;\sin\widetilde{\mathcal{R}}(\lambda),g(\lambda)\}$
from Eqs.\ (\ref{eq: previous result}), (\ref{eq: summarized inequalities one side a}), 
and (\ref{eq: summarized inequalities one side b}).
\label{Fig: outline}
         }
\end{center}
\end{figure}

{\bf Derivation of improved inequalities.---}
In this section, we develop an approach involving two orthonormal vectors to derive inequalities that are stronger than 
the one in Eq.\ (\ref{eq: previous result}).
Observe that for every given $\lambda$, there are three state vectors involved (see also Fig.\ \ref{Fig: outline}),
i.e., $|\Phi^{\,}_{0}\rangle,$ $|\Psi^{\,}_{\lambda}\rangle,$ and $|\Phi^{\,}_{\lambda}\rangle.$
Of which, only $|\Phi^{\,}_{0}\rangle$ is time independent and is still present at a different value of $\lambda$.
Therefore, a natural strategy is to decompose the other two states, $|\Psi^{\,}_{\lambda}\rangle$ and $|\Phi^{\,}_{\lambda}\rangle$, 
into the initial ground state $|\Phi^{\,}_{0}\rangle$ and its orthogonal complement
(think of the Gram-Schmidt process).
Let $|\Phi^{\perp}_{0}(\lambda)\rangle$ be a $\lambda$-dependent normalized state that is orthogonal to
the initial state $|\Phi^{\,}_{0}\rangle,$ i.e., $\langle\Phi^{\,}_{0}|\Phi^{\perp}_{0}(\lambda)\rangle=0,$
we then decompose the physical state $|\Psi^{\,}_{\lambda}\rangle$ in terms of these two orthonormal states,
\begin{align}
|\Psi^{\,}_{\lambda}\rangle=e^{\mathrm{i}\varphi^{\,}_{\lambda}}\cos\theta^{\,}_{\lambda}|\Phi^{\,}_{0}\rangle
+\sin\theta^{\,}_{\lambda}|\Phi^{\perp}_{0}(\lambda)\rangle,
\label{eq: psi expansion}
\end{align}
where $\theta^{\,}_{\lambda}\in[0,\pi/2]$ and $\varphi^{\,}_{\lambda}\in[0,2\pi]$ with the subscript $\lambda$ indicates that both $\theta^{\,}_{\lambda}$ 
and $\varphi^{\,}_{\lambda}$ are a function of $\lambda=\lambda(t)$.
Notice that, by construction,
\begin{subequations}
\label{eq: properties by construction}
\begin{align}
&
|\langle\Phi^{\,}_{0}|\Psi^{\,}_{\lambda}\rangle| = \cos\theta^{\,}_{\lambda}
\quad\Leftrightarrow\quad
\theta^{\,}_{\lambda}=\frac{\pi}{2}D(\Psi^{\,}_{\lambda},\Psi^{\,}_{0}),
\label{eq: properties by construction a}
\\
&
\langle\Phi^{\perp}_{0}(\lambda)|\Psi^{\,}_{\lambda}\rangle = \sin\theta^{\,}_{\lambda}=\sin\left(\frac{\pi}{2}D(\Psi^{\,}_{\lambda},\Psi^{\,}_{0})\right).
\label{eq: properties by construction b}
\end{align} 
\end{subequations}
Similarly, the instantaneous ground state $|\Phi^{\,}_{\lambda}\rangle$ can be decomposed into the initial state $|\Phi^{\,}_{0}\rangle$ and 
another $\lambda$-dependent orthogonal complement $|\widetilde{\Phi}^{\perp}_{0}(\lambda)\rangle$
[which need not be the same as $|\Phi^{\perp}_{0}(\lambda)\rangle$ introduced in Eq.\ (\ref{eq: psi expansion})],
\begin{align}
\label{eq: phi expansion}
|\Phi^{\,}_{\lambda}\rangle=
\langle\Phi^{\,}_{0}|\Phi^{\,}_{\lambda}\rangle|\Phi^{\,}_{0}\rangle
+
\langle\widetilde{\Phi}^{\perp}_{0}(\lambda)|\Phi^{\,}_{\lambda}\rangle|\widetilde{\Phi}^{\perp}_{0}(\lambda)\rangle.
\end{align}
The normalization condition, 
$1=\langle\Phi^{\,}_{\lambda}|\Phi^{\,}_{\lambda}\rangle$,
then implies
\begin{align}
|\langle\widetilde{\Phi}^{\perp}_{0}(\lambda)|\Phi^{\,}_{\lambda}\rangle|
=\sqrt{1-|\langle\Phi^{\,}_{0}|\Phi^{\,}_{\lambda}\rangle|^{2}}
=
\sqrt{1-\mathcal{C}(\lambda)},
\label{eq: define square of one minus capitalC}
\end{align}
where $\mathcal{C}(\lambda)$ is defined in Eq.\ (\ref{eq: orthogonality catastrophe}).

Since the components of the physical state $|\Psi^{\,}_{\lambda}\rangle$ (\ref{eq: psi expansion}) are entirely determined by the Bures angle, 
$D(\Psi^{\,}_{\lambda},\Psi^{\,}_{0})=(2/\pi)\theta^{\,}_{\lambda}$, 
and that of the instantaneous ground state $|\Phi^{\,}_{\lambda}\rangle$ (\ref{eq: phi expansion}) by the generalized orthogonality catastrophe, 
$\mathcal{C}(\lambda),$ 
it is then obvious that their overlap, the adiabatic fidelity $\mathcal{F}(\lambda)$ (\ref{eq: adiabatic fidelity}), 
should be wholly determined by both $\theta^{\,}_{\lambda}$ and $\mathcal{C}(\lambda)$,
as will be seen shortly.

We are in a position to 
compute the adiabatic fidelity $\mathcal{F}(\lambda)$
(\ref{eq: adiabatic fidelity}) 
using Eq.\ (\ref{eq: psi expansion}),
\begin{align}
\mathcal{F}(\lambda)
=
\cos^2\theta^{\,}_{\lambda}|\langle\Phi^{\,}_{\lambda}|\Phi^{\,}_{0}\rangle|^{2}
+
\sin^2\theta^{\,}_{\lambda}|\langle\Phi^{\,}_{\lambda}|\Phi^{\perp}_{0}(\lambda)\rangle|^{2}
\nonumber\\
+
\Re
\Big(
e^{\mathrm{i}\varphi^{\,}_{\lambda}}\sin(2\theta^{\,}_{\lambda})\langle\Phi^{\,}_{\lambda}|\Phi^{\,}_{0}\rangle\langle\Phi^{\perp}_{0}(\lambda)|\Phi^{\,}_{\lambda}\rangle
\Big).
\label{eq: F adb in perp basis}
\end{align}
The main object of interest, $|\mathcal{F}(\lambda)-\mathcal{C}(\lambda)|,$
can then be computed  
using (i) Eq.\ (\ref{eq: F adb in perp basis}),
(ii) the triangle inequality for absolute value, 
and
(iii) the inequality $\Re(z)\leq|z|$ for $z\in\mathbb{C}$,
\begin{align}
&
\left|
\mathcal{F}(\lambda)
-
\mathcal{C}(\lambda)
\right|
\nonumber\\
\stackrel{\rm (i)}{=}&\,
\Big|
\sin^2\theta^{\,}_{\lambda}
\left(
-
|\langle\Phi^{\,}_{\lambda}|\Phi^{\,}_{0}\rangle|^{2}
+
|\langle\Phi^{\,}_{\lambda}|\Phi^{\perp}_{0}(\lambda)\rangle|^{2}
\right)
\nonumber\\
&+
\Re
\Big(
e^{\mathrm{i}\varphi^{\,}_{\lambda}}\sin(2\theta^{\,}_{\lambda})\langle\Phi^{\,}_{\lambda}|\Phi^{\,}_{0}\rangle\langle\Phi^{\perp}_{0}(\lambda)|\Phi^{\,}_{\lambda}\rangle
\Big)
\Big|
\nonumber\\
\stackrel{\rm (ii)}{\leq}&\,
\Big|
\sin^2\theta^{\,}_{\lambda}
\left(
-
|\langle\Phi^{\,}_{\lambda}|\Phi^{\,}_{0}\rangle|^{2}
+
|\langle\Phi^{\,}_{\lambda}|\Phi^{\perp}_{0}(\lambda)\rangle|^{2}
\right)
\Big|
\nonumber\\
&+
\Big|
\Re
\Big(
e^{\mathrm{i}\varphi^{\,}_{\lambda}}\sin(2\theta^{\,}_{\lambda})\langle\Phi^{\,}_{\lambda}|\Phi^{\,}_{0}\rangle\langle\Phi^{\perp}_{0}(\lambda)|\Phi^{\,}_{\lambda}\rangle
\Big)
\Big|
\nonumber\\
\stackrel{\rm (iii)}{\leq}&\,
\Big|
\sin^2\theta^{\,}_{\lambda}
\left(
-
|\langle\Phi^{\,}_{\lambda}|\Phi^{\,}_{0}\rangle|^{2}
+
|\langle\Phi^{\,}_{\lambda}|\Phi^{\perp}_{0}(\lambda)\rangle|^{2}
\right)
\Big|
\nonumber\\
&+
\sin(2\theta^{\,}_{\lambda})
\left|\langle\Phi^{\,}_{\lambda}|\Phi^{\,}_{0}\rangle\right|
\left|\langle\Phi^{\perp}_{0}(\lambda)|\Phi^{\,}_{\lambda}\rangle\right|
\nonumber\\
\leq&\,
\Big|
\sin^2\theta^{\,}_{\lambda}
\left(
-
|\langle\Phi^{\,}_{\lambda}|\Phi^{\,}_{0}\rangle|^{2}
+
|\langle\Phi^{\,}_{\lambda}|\widetilde{\Phi}^{\perp}_{0}(\lambda)\rangle|^{2}
\right)
\Big|
\nonumber\\
&+
\sin(2\theta^{\,}_{\lambda})
\left|\langle\Phi^{\,}_{\lambda}|\Phi^{\,}_{0}\rangle\right|
\left|\langle\widetilde{\Phi}^{\perp}_{0}(\lambda)|\Phi^{\,}_{\lambda}\rangle\right|,
\label{eq: inequality perp basis}
\end{align}
where the last expression is obtained after using the following inequality,
\begin{align}
|\langle \Phi^{\perp}_{0}(\lambda)|\Phi^{\,}_{\lambda}\rangle|
\;\leq\;
|\langle\widetilde{\Phi}^{\perp}_{0}(\lambda)|\Phi^{\,}_{\lambda}\rangle|,
\label{eq: an useful inequality for inner product of two orthogonal complements}
\end{align}
which is a result of Eq.\ (\ref{eq: phi expansion}) 
with $|\langle \Phi^{\perp}_{0}(\lambda)|\widetilde{\Phi}^{\perp}_{0}(\lambda)\rangle|\leq1$.

Making use of Eq.\ (\ref{eq: define square of one minus capitalC}) to express $|\langle\Phi^{\,}_{\lambda}|\Phi^{\,}_{0}\rangle|$ 
and $|\langle\Phi^{\,}_{\lambda}|\widetilde{\Phi}^{\perp}_{0}(\lambda)\rangle|$
in terms of $\sqrt{\mathcal{C}(\lambda)}$ and $\sqrt{1-\mathcal{C}(\lambda)}$, respectively, the inequality (\ref{eq: inequality perp basis}) then reads
\begin{subequations}
\label{eq: inequality perp basis in terms of capitalC all}
\begin{align}
\left|\mathcal{F}(\lambda)-\mathcal{C}(\lambda)\right|
\quad\leq\quad
g(\mathcal{C}(\lambda),\theta^{\,}_{\lambda}),
\label{eq: inequality perp basis in terms of capitalC}
\end{align}
where we have introduced an auxiliary function $g(\mathcal{C},\theta)$ for later convenience,
\begin{align}
g(\mathcal{C},\theta)\:=\sin^2\theta
\left|
1-
2\mathcal{C}
\right|
+
\sin(2\theta)
\sqrt{\mathcal{C}}\sqrt{1-\mathcal{C}}.
\label{eq: define an auxiliary function}
\end{align}
\end{subequations}
Note that the right side of Eq.\ (\ref{eq: inequality perp basis in terms of capitalC}) depends on only two independent variables, i.e.,
$\mathcal{C}(\lambda)$ and $\theta^{\,}_{\lambda},$ as claimed previously.

It remains to find upper bounds on the function $g(\mathcal{C},\theta)$ (\ref{eq: define an auxiliary function}).
To this end, 
treating it as a function of $\mathcal{C}$ alone,
one finds two degenerate global maxima of $g(\mathcal{C},\theta)$ occur when $\mathcal{C}=(1\pm\sin\theta)/2$ and yields
\begin{align}
\max_{0\leq\mathcal{C}\leq1}g(\mathcal{C},\theta)=\sin\theta.
\end{align}
Therefore, an upper bound for the right side of Eq.\ (\ref{eq: inequality perp basis in terms of capitalC}) is obtained as
\begin{align}
\left|\mathcal{F}(\lambda)-\mathcal{C}(\lambda)\right|
\;\leq\;
\max_{0\leq\mathcal{C}(\lambda)\leq1}g(\mathcal{C}(\lambda),\theta^{\,}_{\lambda})
\;=\;
\sin\theta^{\,}_{\lambda}.
\label{eq: new inequality}
\end{align}
Note that the inequality $\left|\mathcal{F}(\lambda)-\mathcal{C}(\lambda)\right|\leq\sin\theta^{\,}_{\lambda}$ (\ref{eq: new inequality})
can also be proved alternatively using a fairly elementary method explained in Supplemental Material {\bf S2},
which seems to first appear in Refs.\ \cite{Rastegin02a,Rastegin02b}.
Upon using the bound from the quantum speed limit (\ref{eq: bound from QSL})
[recall the relation $\theta^{\,}_{\lambda}=\frac{\pi}{2}D(\Psi^{\,}_{\lambda},\Psi^{\,}_{0})$ from Eq.\ (\ref{eq: properties by construction a})]
and considering the fact that $\sin x$ is a monotonically increasing function for $x\in[0,\pi/2],$
the rightmost side in Eq.\ (\ref{eq: new inequality}) can be further bounded from above by $\sin\widetilde{\mathcal{R}}(\lambda),$
\begin{align}
\left|
\mathcal{F}(\lambda)-\mathcal{C}(\lambda)
\right|
\;\leq\;
\sin\theta^{\,}_{\lambda}
\;\leq\;
\sin\widetilde{\mathcal{R}}(\lambda).
\label{eq: summarized inequalities one side a}
\end{align}
This is the first improved inequality mentioned in the Introduction.
It is evident that this inequality (\ref{eq: summarized inequalities one side a})
provides a stronger bound compared to the previous inequality (\ref{eq: previous result})
since $\sin x\leq x$ for $x\in[0,\pi/2].$

One may wonder whether it is possible to obtain an upper bound that is stronger than $\sin\widetilde{\mathcal{R}}(\lambda)$ (\ref{eq: summarized inequalities one side a})
by manipulating the function $g(\mathcal{C}(\lambda),\theta^{\,}_{\lambda})$ defined in Eq.\ (\ref{eq: define an auxiliary function}).
The answer is affirmative
provided an upper bound on the $\sin(2\theta^{\,}_{\lambda})$ term of Eq.\ (\ref{eq: inequality perp basis in terms of capitalC all})
can be found.
To show this, recall that $\theta^{\,}_{\lambda}=\frac{\pi}{2}D(\Psi^{\,}_{\lambda},\Psi^{\,}_{0})$ is upper bounded 
by using the quantum speed limit (\ref{eq: bound from QSL}),
$
\theta^{\,}_{\lambda}
\leq
\widetilde{\mathcal{R}}(\lambda).
$
It then follows that
$\sin\theta^{\,}_{\lambda}\leq\sin\widetilde{\mathcal{R}}(\lambda)$
and
$\sin(2\theta^{\,}_{\lambda})\leq\sin(2\widetilde{\widetilde{\mathcal{R}}}(\lambda))$
for
$\theta^{\,}_{\lambda}\leq
\widetilde{\mathcal{R}}(\lambda)
\leq\pi/2,$
where 
\begin{align}
\widetilde{\widetilde{\mathcal{R}}}(\lambda)
\:=
\min\left(\mathcal{R}(\lambda),\frac{\pi}{4}\right).
\label{eq: define another auxiliary R}
\end{align}
Using these facts, one may further bound the right side of Eq.\ (\ref{eq: inequality perp basis in terms of capitalC}) from above
\begin{align}
\left|\mathcal{F}(\lambda)-\mathcal{C}(\lambda)\right|
&\leq
g(\lambda),
\label{eq: summarized inequalities one side b}
\end{align}
where the function $g(\lambda)$ reads
\begin{subequations}
\label{eq: define g as a function of lambda}
\begin{align}
&g(\lambda)
\:=
g^{\,}_{1}(\lambda)
+
g^{\,}_{2}(\lambda),
\label{eq: define g as a function of lambda a}
\\
&g^{\,}_{1}(\lambda)
\:=\sin^2\widetilde{\mathcal{R}}(\lambda)
\left|
1-
2\mathcal{C}(\lambda)
\right|,
\label{eq: define g as a function of lambda b}
\\
&g^{\,}_{2}(\lambda)\:=
\sin(2\widetilde{\widetilde{\mathcal{R}}}(\lambda))
\sqrt{\mathcal{C}(\lambda)}\sqrt{1-\mathcal{C}(\lambda)},
\label{eq: define g as a function of lambda c}
\end{align}
\end{subequations}
where $\widetilde{\mathcal{R}}(\lambda)$ and $\widetilde{\widetilde{\mathcal{R}}}(\lambda)$ 
are defined in Eqs.\ (\ref{eq: bound from QSL a}) and (\ref{eq: define another auxiliary R}), respectively.
The inequality (\ref{eq: summarized inequalities one side b}) 
is the second improved inequality mentioned in the Introduction.

We want to  
emphasize that the two improved inequalities, Eqs.\ (\ref{eq: summarized inequalities one side a})
and (\ref{eq: summarized inequalities one side b}),
are applicable to any quantum system, 
no matter whether the system size is large or small.
Nevertheless, as is demonstrated in a later section, 
the second improved inequality (\ref{eq: summarized inequalities one side b}) 
is particularly powerful when the system size is large.

\begin{figure*}[t]
\begin{center}
(a)
\includegraphics[width=0.4\textwidth]{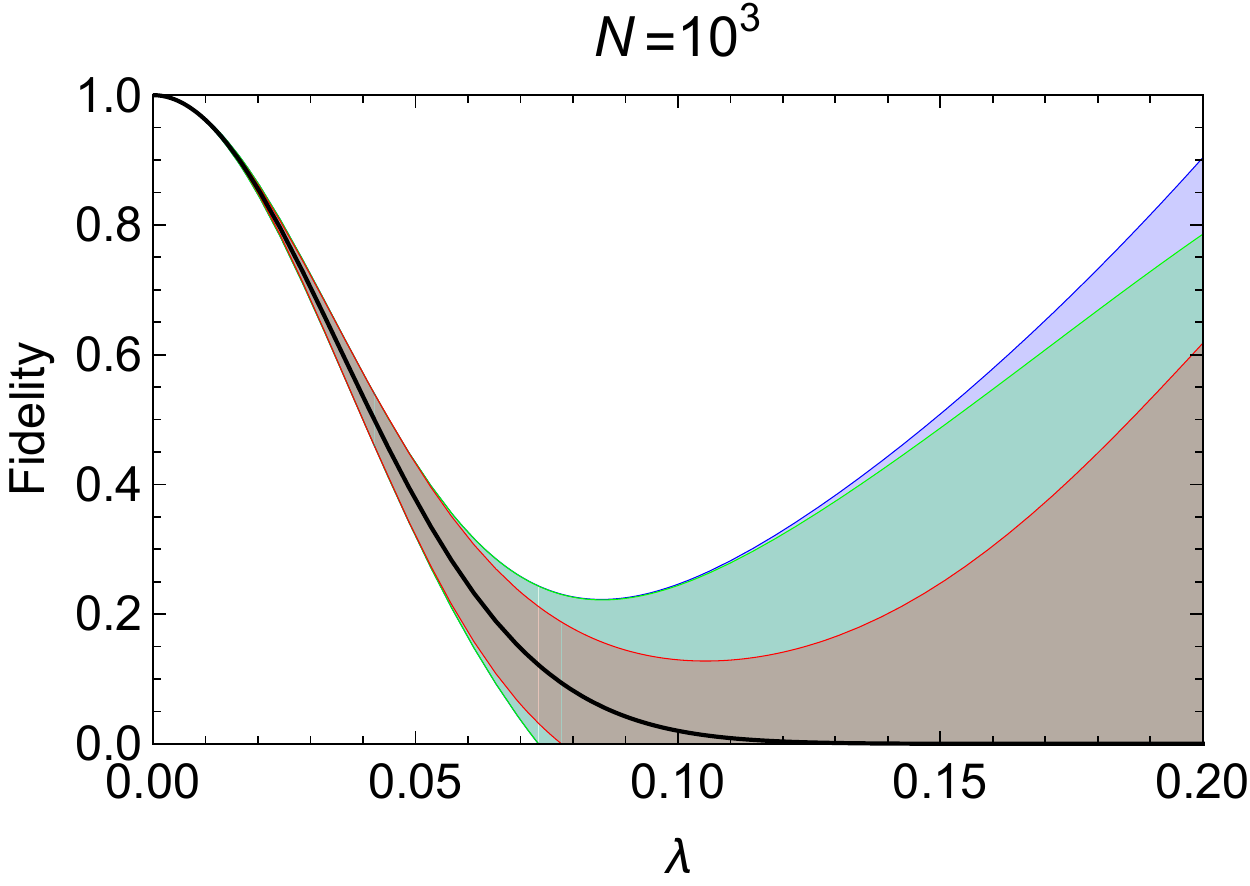}
(b)
\includegraphics[width=0.51\textwidth]{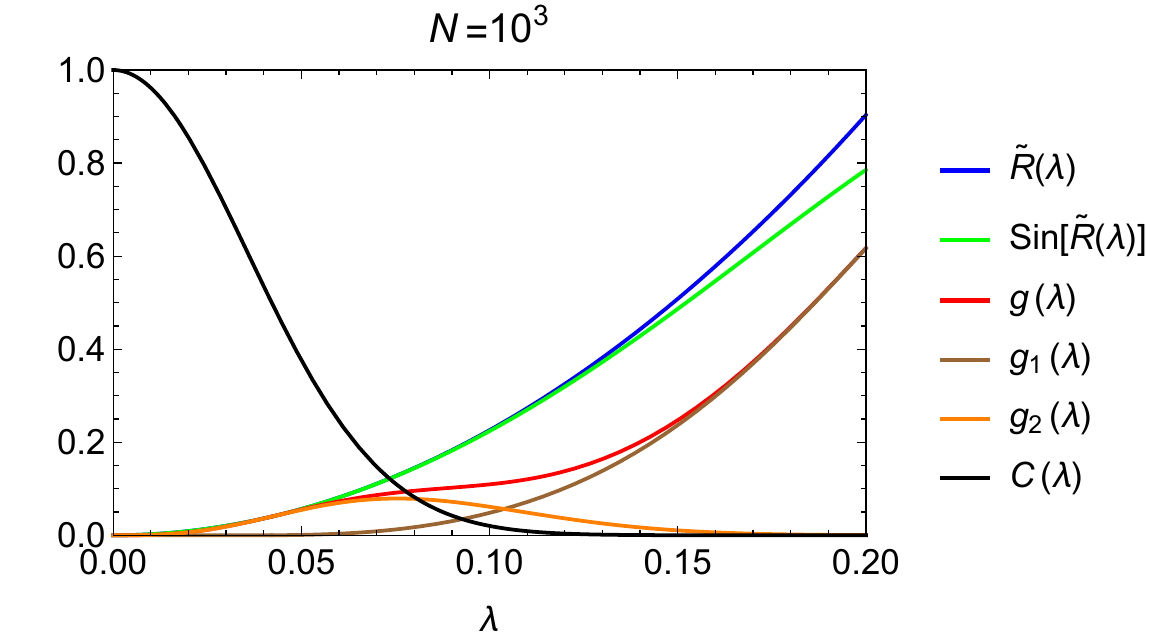}
\caption{
(Color online)
Adiabatic fidelity $\mathcal{F}(\lambda)$ and generalized orthogonality catastrophe $\mathcal{C}(\lambda)$ described by 
the Hamiltonian (\ref{eq: RM model}) with the parametrization 
(\ref{eq: chosen value of parameters}) for $N=10^3$ sites.
(a)
Comparison between the old inequality (\ref{eq: previous result}) 
and the two improved inequalities (\ref{eq: summarized inequalities one side a}) and (\ref{eq: summarized inequalities one side b}).
The black curve is for $\mathcal{C}(\lambda),$ which is, however, indistinguishable from $\mathcal{F}(\lambda)$ in the plot.
The blue-shaded (resp., green- and red-shaded) region is the bound for 
$\mathcal{F}(\lambda)$ 
from $\widetilde{\mathcal{R}}(\lambda)$ (\ref{eq: previous result}) 
[resp., from $\sin\widetilde{\mathcal{R}}(\lambda)$ in Eq.\ (\ref{eq: summarized inequalities one side a}) 
and from $g(\lambda)$ in Eq.\ (\ref{eq: summarized inequalities one side b})].
The red-shaded (resp., green-shaded) area is about 59\% (resp., 95\%) of the blue-shaded area.
(b)
Behavior of the functions
$\widetilde{\mathcal{R}}(\lambda)$, $\sin\widetilde{\mathcal{R}}(\lambda)$,
and $g(\lambda)=g^{\,}_{1}(\lambda)+g^{\,}_{2}(\lambda)$ (\ref{eq: define g as a function of lambda}) as a function of $\lambda.$
For comparison, $\mathcal{C}(\lambda)$ is also depicted.
Refer to the main text for further explanation.
\label{Fig: fidelity}
         }
\end{center}
\end{figure*}

{\bf Setup of driven many-body systems.---}
Before considering a specific example in the next section,
we follow Ref.\ \cite{Lychkovskiy17} to specify a wide range of quantum systems that share general properties which the specific example possesses.
The time-dependent Hamiltonian $H^{\,}_{\lambda}$ in which we are interested 
has a typical form,
\begin{align}
H^{\,}_{\lambda} = H^{\,}_{0} + \lambda V,
\end{align}
where $H^{\,}_{0}$ is a time-independent Hamiltonian with the lowest energy eigenstate $|\Phi^{\,}_{0}\rangle$ and 
$
V
$
is a driving potential.
We also assume that the driving rate $\Gamma=\partial^{\,}_{t}\lambda$ is a constant in $\lambda.$
It then follows that $\mathcal{R}(\lambda)$ [(\ref{eq: bound from QSL b})],
the time integral of quantum uncertainty,
reads
\begin{align}
\mathcal{R}(\lambda)&=
\frac{\lambda^2}{2\Gamma}
\delta V^{\,}_{N},
\quad
\delta V^{\,}_{N}
\:=
\sqrt{
\langle\Phi^{\,}_{0}|V^2|\Phi^{\,}_{0}\rangle
-
\langle\Phi^{\,}_{0}|V|\Phi^{\,}_{0}\rangle^{2}
}.
\label{eq: define std of V}
\end{align}
In other words, $\mathcal{R}(\lambda)$ is a monotonically increasing function in $\lambda^2$.

We further restrict ourselves to a broad class of time-dependent Hamiltonians
whose generalized orthogonality catastrophe $\mathcal{C}(\lambda)$ (\ref{eq: orthogonality catastrophe}) has the following simple exponentially decaying form
when the system size, $N$, is large,
\begin{align}
\ln\mathcal{C}(\lambda) = -C^{\,}_{N}\lambda^2+r(N,\lambda),
\quad
\lim_{N\to\infinity}C^{\,}_{N}=\infinity,
\label{eq: orthogonality catastrophe asymptotic}
\end{align}
where the residual $r$ satisfies $\lim^{\,}_{N\to\infinity}r(N,C^{-1/2}_{N})=0.$

{\bf Example: driven Rice-Mele model.---}
In order to demonstrate the validity of the improved inequalities, Eqs.\ (\ref{eq: summarized inequalities one side a}) and (\ref{eq: summarized inequalities one side b}),
we consider the spinless Rice-Mele model on a half-filled one-dimensional bipartite lattice with the Hamiltonian
\cite{Rice82,Nakajima16}
\begin{align}
\label{eq: RM model}
H^{\,}_{\mathrm{RM}}
=&
\sum^{N}_{j=1}
\left[
-(J+U)a^{\dag}_{j}b^{\,}_{j}
-(J-U)a^{\dag}_{j}b^{\,}_{j+1}
+\mathrm{h.c.}
\right]
\nonumber\\
&
+\sum^{N}_{j=1}
\Delta
\left(
a^{\dag}_{j}a^{\,}_{j}
-
b^{\dag}_{j}b^{\,}_{j}
\right),
\end{align}
where $a^{\,}_{j}$ and $b^{\,}_{j}$ are the fermion annihilation operators on the $a$ and $b$ sublattices, respectively.
Here, $N$ is the number of lattice sites.
For the case of $J=U=\mathrm{const}$ and $\Delta=\lambda E^{\,}_{\mathrm{R}},$ where $\lambda=\Gamma t$ and $E^{\,}_{\mathrm{R}}$ is a recoil energy,
it is shown in Ref.\ \cite{Lychkovskiy17} that
the exponent $C^{\,}_{N}$ defined in Eq.\ (\ref{eq: orthogonality catastrophe asymptotic}) and 
the quantum uncertainty $\delta V^{\,}_{N}$ defined in Eq.\ (\ref{eq: define std of V}) read as follows:
\begin{align}
C^{\,}_{N}=\frac{NE^{2}_{\mathrm{R}}}{16JU},
\quad
\delta V^{\,}_{N}=\sqrt{N}E^{\,}_{\mathrm{R}}.
\label{eq: parametrization for RM model}
\end{align}
Equation (\ref{eq: parametrization for RM model}) 
with the chosen value of the parameters from Ref.\ \cite{Lychkovskiy17},
\begin{align}
\label{eq: chosen value of parameters}
(J,\;U,\;\Delta,\;\Gamma)=(0.4E^{\,}_{\mathrm{R}},\;0.4E^{\,}_{\mathrm{R}},\;\lambda E^{\,}_{\mathrm{R}},\;0.7E^{\,}_{\mathrm{R}}),
\end{align}
gives the following expressions for
$\mathcal{C}(\lambda)$ (\ref{eq: orthogonality catastrophe asymptotic})
and 
$\mathcal{R}(\lambda)$ (\ref{eq: define std of V}):
\begin{align}
\label{eq: parametrization for RM model capitalC and capitalR}
\mathcal{C}(\lambda)=e^{-N\lambda^2/(1.6)^2},
\quad
\mathcal{R}(\lambda)=\frac{\sqrt{N}\lambda^{2}}{1.4}.
\end{align}

Provided with Eq.\ (\ref{eq: parametrization for RM model capitalC and capitalR}), 
we present in Fig.\ \ref{Fig: fidelity} the comparison of bounds on the adiabatic fidelity $\mathcal{F}(\lambda)$ 
using the old inequality (\ref{eq: previous result}) and the two improved inequalities, Eqs.\ (\ref{eq: summarized inequalities one side a}) 
and (\ref{eq: summarized inequalities one side b}),
for $N=10^3.$
Specifically, 
given the second improved inequality (\ref{eq: summarized inequalities one side b})
and noting that $\mathcal{F}(\lambda)\in[0,1]$ by its definition,
the following two-sided bound on the adiabatic fidelity $\mathcal{F}(\lambda)$ is obtained,
\begin{align}
\max\left(\mathcal{C}(\lambda)-g(\lambda),\;0\right)
\leq
\mathcal{F}(\lambda)
\leq
\min\left(\mathcal{C}(\lambda)+g(\lambda),\;1\right).
\nonumber
\end{align}
Similar expressions apply to the old inequality (\ref{eq: previous result}) and the first improved inequality (\ref{eq: summarized inequalities one side a})
with $g(\lambda)$ being replaced by $\widetilde{\mathcal{R}}(\lambda)$ and $\sin\widetilde{\mathcal{R}}(\lambda)$, respectively.
Figure \ref{Fig: fidelity}(a) shows that the second improved inequality (\ref{eq: summarized inequalities one side b})
(as represented by the red-shaded region)
greatly improves the estimate for 
$\mathcal{F}(\lambda)$ compared to the previous estimate (\ref{eq: previous result})
(as represented by the blue-shaded region). 
Figure \ref{Fig: fidelity}(b)
shows the behavior of functions
$\widetilde{\mathcal{R}}(\lambda)$, $\sin\widetilde{\mathcal{R}}(\lambda)$,
and $g(\lambda)=g^{\,}_{1}(\lambda)+g^{\,}_{2}(\lambda)$ (\ref{eq: define g as a function of lambda}) as a function of $\lambda.$
The function
$g(\lambda)$
is dominated by $g^{\,}_{2}(\lambda)$ when $\lambda$ is small and is dominated by $g^{\,}_{1}(\lambda)$ when $\lambda$ is large.
It can be understood from Eq.\ (\ref{eq: define g as a function of lambda}) that, for large $\lambda,$ 
$\mathcal{C}(\lambda)\approx0,$
so that $g(\lambda)\approx g^{\,}_{1}(\lambda)\approx\sin^2\mathcal{R}(\lambda)$.
Similarly, for small $\lambda$, $\sin(2\mathcal{R}(\lambda))\gg\sin^2\mathcal{R}(\lambda),$ so that $g(\lambda)\approx g^{\,}_{2}(\lambda).$

Clearly, the upper boundary of each shaded region in Fig.\ \ref{Fig: fidelity}(a) is determined by 
$\mathcal{C}(\lambda)+\widetilde{\mathcal{R}}(\lambda)$ for the blue-shaded region, 
by $\mathcal{C}(\lambda)+\sin\widetilde{\mathcal{R}}(\lambda)$ for the green-shaded region, 
and by $\mathcal{C}(\lambda)+g(\lambda)$ for the red-shaded region.
Observe from 
Fig.\ \ref{Fig: fidelity}(b)
that the function $\mathcal{C}(\lambda)$ is an exponentially decaying function in $\lambda,$ whereas the functions 
$\widetilde{\mathcal{R}}(\lambda)$, $\sin\widetilde{\mathcal{R}}(\lambda)$, and $g(\lambda)$ 
are monotonically increasing function in $\lambda$. 
As a result, the upper boundary of each shaded region in Fig.\ \ref{Fig: fidelity}(a)
has a valley when $\mathcal{C}(\lambda)$ is too small
and then monotonically increases as $\lambda$ increases.
Similarly, the bottom boundary of each shaded region in Fig.\ \ref{Fig: fidelity}(a) is determined by 
$\mathcal{C}(\lambda)-\widetilde{\mathcal{R}}(\lambda)$, $\mathcal{C}(\lambda)-\sin\widetilde{\mathcal{R}}(\lambda)$, and
$\mathcal{C}(\lambda)-g(\lambda),$ respectively.
The bottom boundary is at zero when 
$\mathcal{C}(\lambda)\leq\left\{\widetilde{\mathcal{R}}(\lambda),\;\sin\widetilde{\mathcal{R}}(\lambda),\;g(\lambda)\right\}.$

To further investigate the effect of increasing system size on the bounds for the
adiabatic fidelity
$\mathcal{F}(\lambda)$, 
we plot in Fig.\ \ref{Fig: fidelity large N}
the cases of $N=10^4$ and $N=10^5.$
For both cases, the green-shaded area is almost identical to the blue-shaded area, whereas
the red-shaded area is about 34\% (resp., 23\%) 
of the blue-shaded area for $N=10^4$ (resp, for $N=10^5$).
This indicates that the second improved inequality (\ref{eq: summarized inequalities one side b})
is much stronger than the old inequality (\ref{eq: previous result}) as the number of lattice sites, $N$, increases. 
This fact can be understood by noticing that 
$\mathcal{C}(\lambda)$ (\ref{eq: parametrization for RM model capitalC and capitalR}),
the generalized orthogonality catastrophe,
decays quicker with increasing $N$;
we are then forced to concentrate on the region of smaller $\lambda$. 
Consequently, it renders a smaller $g(\lambda)$, 
so that the bounds, $\mathcal{C}(\lambda)\pm g(\lambda)$, are tighter as $N$ increases.

{\bf Implication on adiabaticity breakdown.---}
This section discusses an implication from the second improved inequality (\ref{eq: summarized inequalities one side b}) on the time scale for adiabaticity breakdown 
in generic driven many-body systems that possess the following asymptotic property,
\begin{align}
\label{eq: general assumption on driven models}
\lim^{\,}_{N\to\infinity} \frac{\delta V^{\,}_{N}}{C^{\,}_{N}}=0,
\end{align}
where $\delta V^{\,}_{N}$ is the quantum uncertainty of the driving potential (\ref{eq: define std of V})
and $C^{\,}_{N}$ is the exponent of the generalized orthogonality catastrophe (\ref{eq: orthogonality catastrophe asymptotic}).
As was pointed out in Ref.\ \cite{Lychkovskiy17} using a general scaling argument,
the asymptotic form (\ref{eq: general assumption on driven models}) is obeyed by a wide range of Hamiltonians
\footnote{
For example, for both gapped and gapless systems in $D$-dimensional space with a bulk driving ($d=D$) or a boundary driving $(d=D-1)$,
the scaling form of $\delta V^{\,}_{N}$ and $C^{\,}_{N}$ reads
$(\delta V^{\,}_{N}, C^{\,}_{N}) \sim (N^{d/(2D)},N^{d/D})$,
while the scaling form is $(\delta V^{\,}_{N}, C^{\,}_{N}) \sim (1,\log N)$ for gapless systems with a
local driving on a single space point. 
}.

Now, combining the adiabaticity condition (\ref{thm: QAT}) and the bound on the
adiabatic fidelity $\mathcal{F}(\lambda)$ from the second improved inequality 
(\ref{eq: summarized inequalities one side b}) yields
\begin{align}
1-\epsilon-\mathcal{C}(\lambda)\leq g(\lambda).
\label{eq: adiabaticity condition with bounds}
\end{align}
This inequality has to be satisfied if adiabaticity is established.

We follow the discussion of Ref.\ \cite{Lychkovskiy17} to define the adiabatic mean free path $\lambda^{\,}_{*}$ as a solution to $\mathcal{F}(\lambda^{\,}_{*})=1/e.$
The leading asymptotic of $\lambda^{\,}_{*}$ reads
\begin{align}
\lambda^{\,}_{*}=C^{-1/2}_{N}.
\label{eq: adiabatic mean free path}
\end{align}
Observe that if the driving rate $\Gamma$ is independent of the system size $N$,
then Eq.\ (\ref{eq: general assumption on driven models}) 
indicates that $\mathcal{R}(\lambda=\lambda^{\,}_{*})$ (\ref{eq: define std of V}) 
vanishes under the limit of large $N$,
\begin{align}
\lim_{N\to\infinity}\mathcal{R}(\lambda^{\,}_{*})
=
\frac{1}{2\Gamma}
\lim_{N\to\infinity}
\frac{\delta V^{\,}_{N}}{C^{\,}_{N}}
=0.
\label{eq: QSL at diabatic mean free path}
\end{align}
Consequently, under the same large $N$ limit, the asymptotic behavior (\ref{eq: QSL at diabatic mean free path}) causes the function $g(\lambda)$ 
(\ref{eq: define g as a function of lambda}) to
vanish when $\lambda\geq\lambda^{\,}_{*}$.
If so, the inequality (\ref{eq: adiabaticity condition with bounds}) with $\lambda=\lambda^{\,}_{*}$ reads
$
1-\epsilon-e^{-1}\leq0
$
as
$
N\to\infinity.
$
This means $\epsilon$ cannot be arbitrarily small; thus, adiabaticity fails.

In order to avoid the adiabaticity breakdown, one has to allow the driving rate $\Gamma$ to scale down with increasing system size $N$, $\Gamma=\Gamma^{\,}_{N}.$
Here, $\Gamma^{\,}_{N}$ is determined by setting $\lambda=\lambda^{\,}_{*}$ in the inequality (\ref{eq: adiabaticity condition with bounds}) 
and approximating $\sin\mathcal{R}(\lambda^{\,}_{*})\approx\mathcal{R}(\lambda^{\,}_{*}).$ One finds
\begin{align}
\label{eq: scaling of Gamma}
\Gamma^{\,}_{N}\leq\frac{1}{2}\frac{\delta V^{\,}_{N}}{C^{\,}_{N}}
\frac{1}{1-\epsilon-e^{-1}}
M,
\end{align}
where the multiplicative factor 
$M=
e^{-1/2}\left(1-e^{-1}\right)^{1/2}+\left(1-e^{-1}-e^{-2}\right)^{1/2}
\approx1.187.$
Note that applying the same reasoning to the old inequality (\ref{eq: previous result}) delivers the multiplicative factor $M=1$ in Eq.\ (\ref{eq: scaling of Gamma}),
as was shown in Refs.\ \cite{Lychkovskiy17,Lychkovskiy17b}.
That is to say, compared to the old inequality (\ref{eq: previous result}), 
the improved inequality (\ref{eq: summarized inequalities one side b}) does not affect the scaling form of the driving rate $\Gamma^{\,}_{N}$,
but merely increases the multiplicative constant.

\begin{figure}[t]
\begin{center}
(a)
\includegraphics[width=0.4\textwidth]{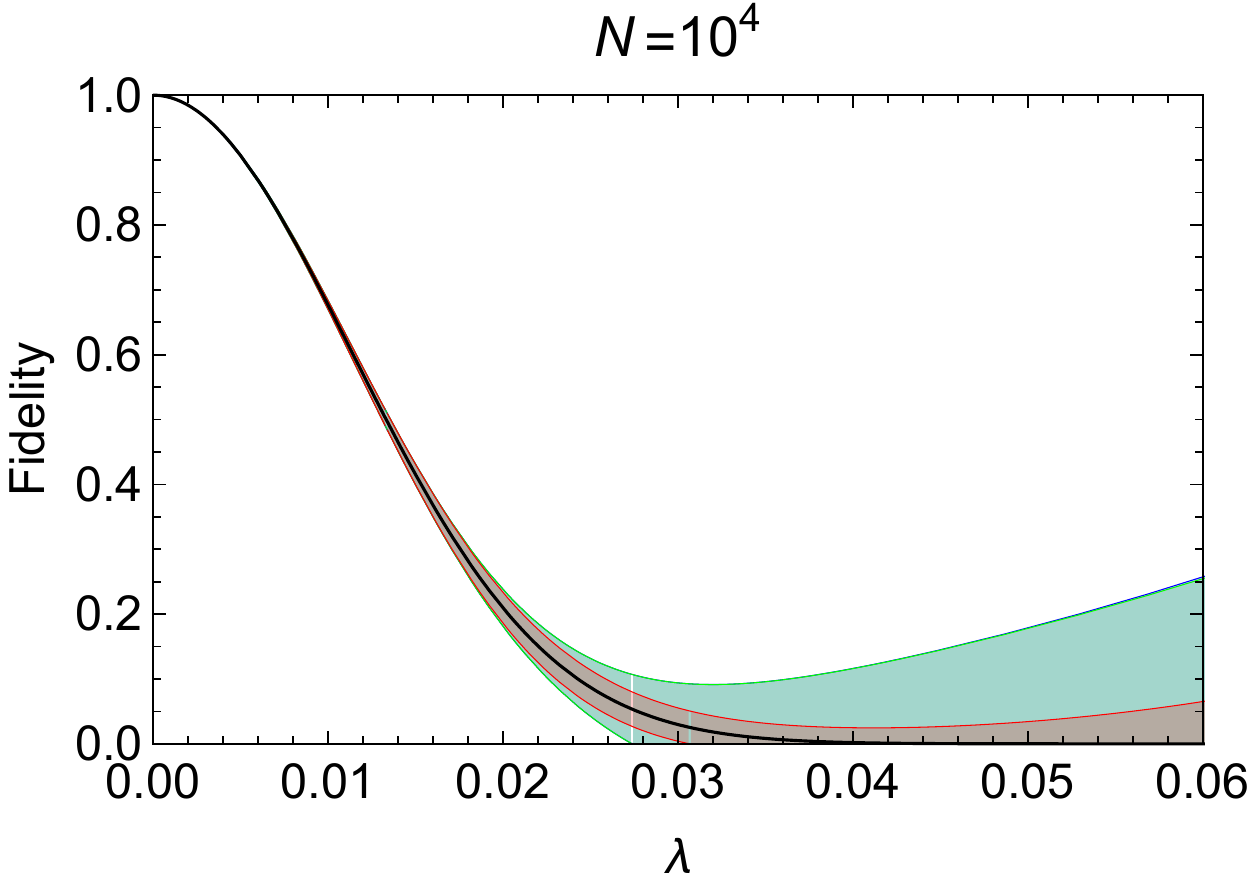}
\\
(b)
\includegraphics[width=0.4\textwidth]{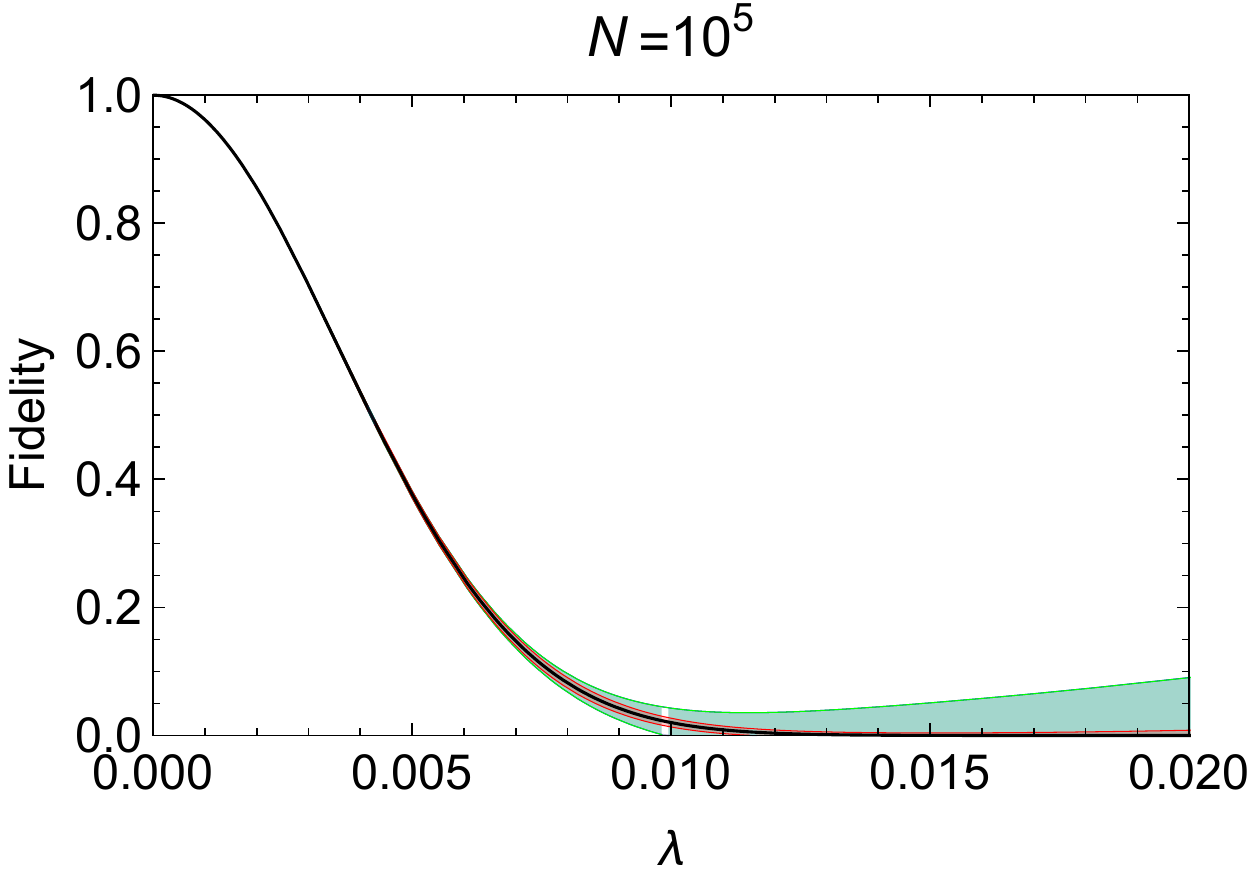}
\caption{
(Color online)
Same as in Fig.\ \ref{Fig: fidelity} but with $N=10^4$ in panel (a) and $N=10^5$ in panel (b).
\label{Fig: fidelity large N}
         }
\end{center}
\end{figure}

{\bf Summary and outlook.---}
In conclusion, we have derived two improved inequalities to bound the adiabatic fidelity using generalized orthogonality catastrophe and the quantum speed limit. 
These two inequalities are stronger than the previous result and are applicable to any quantum system.
In particular, one of the two improved inequalities is nearly sharp when the system size is large. 

In addition to quantum many-body systems, 
our method
could also be applied to other fields
in which bounds on adiabatic evolution are important,
such as adiabatic quantum computation \cite{Farhi00,Roland02,Albash18,Lychkovskiy18,Suzuki20}
and adiabatic quantum control \cite{Rosenfeld96,Leghtas11,Brif14,Meister14,Augier18}. 


{\it Acknowledgments.} 
This work is part of the project Adiabatic Protocols in Extended Quantum Systems, Project No 680-91-130, 
which is funded by the Dutch Research Council (NWO).

\bibliography{references-adiabatic}

\clearpage

\onecolumngrid


\renewcommand{\theequation}{S\arabic{equation}}
\setcounter{equation}{0}
\renewcommand{\thefigure}{S\arabic{figure}}
\setcounter{figure}{0}

\renewcommand{\thesection}{S\arabic{section}}
\setcounter{section}{0}

\begin{center}
\textbf{\large Supplemental Material: Bounds on quantum adiabaticity in driven many-body systems
from generalized orthogonality catastrophe and quantum speed limit}\\
\vspace{0.5cm}
Jyong-Hao Chen$^{1}$ and Vadim Cheianov$^{1}$\\
\vspace{0.2cm}
\textit{$^1$Instituut-Lorentz, Universiteit Leiden, P.O.\ Box 9506, 2300 RA Leiden, The Netherlands}\\
\end{center}


\section{Derivation of the quantum speed limit (\ref{eq: bound from QSL})}
\label{sec: Derivation of the quantum speed limit}

This section provides an alternative derivation for the inequality of quantum speed limit (\ref{eq: bound from QSL}).
Our approach described below, inspired by Ref.\ \cite{Vaidman92}, is pretty elementary as compared with the original proof
given in Refs.\ \cite{Pfeifer93a,Pfeifer95}
(see also the supplemental material of Ref.\ \cite{Lychkovskiy17}).

We start by mentioning a simple formula.
For any Hermitian operator $A$ and any quantum state $|\psi\rangle,$
there is a decomposition for $A|\psi\rangle$ taking the following form \cite{Aharonov90}
\begin{align}
A|\psi\rangle = \langle A\rangle|\psi\rangle + \Delta A|\psi^{\,}_{\perp}\rangle, \qquad \langle\psi|\psi^{\,}_{\perp}\rangle=0,
\label{eq: a key formula}
\end{align}
where $|\psi^{\,}_{\perp}\rangle$ is a vector orthogonal to $|\psi\rangle$, 
$\langle A\rangle = \langle\psi|A|\psi\rangle$ and $\Delta A=\sqrt{\langle A^2\rangle-\langle A\rangle^2}.$

We are given a time-dependent Hamiltonian $H^{\,}_{\lambda(t)}$ and the scaled time-dependent Schr\"odinger equation (\ref{eq: time-dependent SE}).
The first step is to calculate the rate of the change of the fidelity 
$F(\Psi^{\,}_{\lambda},\Psi^{\,}_{0})$ (\ref{eq: define fidelity})
and use the 
scaled time-dependent Schr\"odinger equation (\ref{eq: time-dependent SE}) 
($\hbar$ is restored in this section),
\begin{align}
\frac{\partial}{\partial\lambda}F(\Psi^{\,}_{\lambda},\Psi^{\,}_{0})
=
\frac{\partial}{\partial\lambda}
\Big(
\langle\Psi^{\,}_{\lambda}|\Psi^{\,}_{0}\rangle\langle\Psi^{\,}_{0}|\Psi^{\,}_{\lambda}\rangle
\Big)
=
2\Re\left(
\left(\frac{\partial}{\partial\lambda}
\langle\Psi^{\,}_{\lambda}|\right)|\Psi^{\,}_{0}\rangle\langle\Psi^{\,}_{0}|\Psi^{\,}_{\lambda}\rangle
\right)
=
\frac{2}{\hbar\,\partial^{\,}_{t}\lambda}
\Re
\Big(
\mathrm{i}\langle\Psi^{\,}_{\lambda}|H^{\,}_{\lambda}|\Psi^{\,}_{0}\rangle\langle\Psi^{\,}_{0}|\Psi^{\,}_{\lambda}\rangle
\Big).
\label{eq: manipulate the rate of change}
\end{align}
Next, the most crucial step is to decompose $H^{\,}_{\lambda}|\Psi^{\,}_{0}\rangle$
in Eq.\ (\ref{eq: manipulate the rate of change}) using the formula (\ref{eq: a key formula}),
\begin{align}
H^{\,}_{\lambda}|\Psi^{\,}_{0}\rangle=\langle H^{\,}_{\lambda}\rangle^{\,}_{0}|\Psi^{\,}_{0}\rangle + \Delta E^{\,}_{0}(\lambda)|\Psi^{\perp}_{0}\rangle,
\label{eq: key formula special case}
\end{align}
where $\langle H^{\,}_{\lambda}\rangle^{\,}_{0}=\langle\Psi^{\,}_{0}|H^{\,}_{\lambda}|\Psi^{\,}_{0}\rangle$ and
$\Delta E^{\,}_{0}(\lambda)$ is defined in Eq.\ (\ref{eq: bound from QSL c}). 
Substituting this decomposition into Eq.\ (\ref{eq: manipulate the rate of change}) yields
\begin{align}
\frac{\partial}{\partial\lambda}F(\Psi^{\,}_{\lambda},\Psi^{\,}_{0})
=
\frac{2\langle H^{\,}_{\lambda}\rangle^{\,}_{0}}{\hbar\,\partial^{\,}_{t}\lambda}
\underbrace{
\Re\Big(\mathrm{i}
|\langle\Psi^{\,}_{\lambda}|\Psi^{\,}_{0}\rangle|^{2}
\Big)}^{\,}_{=0}
+
\frac{2\Delta E^{\,}_{0}}{\hbar\,\partial^{\,}_{t}\lambda}
\Re\Big(\mathrm{i}
\langle\Psi^{\,}_{\lambda}|\Psi^{\perp}_{0}\rangle\langle\Psi^{\,}_{0}|\Psi^{\,}_{\lambda}\rangle
\Big).
\end{align}
Upon taking absolute value for both sides and using the inequality
$|\Re(z)|\leq|z|$ for $z\in\mathbb{C}$, the above equation reads
\begin{align}
\left|
\frac{\partial}{\partial\lambda}F(\Psi^{\,}_{\lambda},\Psi^{\,}_{0})
\right|
=
\frac{2\Delta E^{\,}_{0}}{\hbar}
\left|
\frac{1}{\partial^{\,}_{t}\lambda}
\Re\Big(\mathrm{i}
\langle\Psi^{\,}_{\lambda}|\Psi^{\perp}_{0}\rangle\langle\Psi^{\,}_{0}|\Psi^{\,}_{\lambda}\rangle
\Big)
\right|
\leq
\frac{2\Delta E^{\,}_{0}}{\hbar|\partial^{\,}_{t}\lambda|}
\left|
\langle\Psi^{\,}_{\lambda}|\Psi^{\perp}_{0}\rangle
\right|
\left|
\langle\Psi^{\,}_{0}|\Psi^{\,}_{\lambda}\rangle
\right|.
\label{eq: rate of change intermediate}
\end{align}

Now, recall that the physical state $|\Psi^{\,}_{\lambda}\rangle$ can be decomposed in terms of $|\Psi^{\,}_{0}\rangle$ and $|\Psi^{\perp}_{0}\rangle$
as in Eq.\ (\ref{eq: psi expansion}),
the normalization condition of $|\Psi^{\,}_{\lambda}\rangle$ then implies
$
|\langle\Psi^{\,}_{\lambda}|\Psi^{\perp}_{0}\rangle|
=
\sqrt{
1-
|\langle\Psi^{\,}_{\lambda}|\Psi^{\,}_{0}\rangle|^{2}
}
=
\sqrt{1-F(\Psi^{\,}_{\lambda},\Psi^{\,}_{0})}.
$
Therefore, Eq.\ (\ref{eq: rate of change intermediate}) can be written as
\begin{align}
\left|
\frac{\partial}{\partial\lambda}F(\Psi^{\,}_{\lambda},\Psi^{\,}_{0})
\right|
\leq
\frac{2\Delta E^{\,}_{0}}{\hbar|\partial^{\,}_{t}\lambda|}
\sqrt{1-F(\Psi^{\,}_{\lambda},\Psi^{\,}_{0})}\sqrt{F(\Psi^{\,}_{\lambda},\Psi^{\,}_{0})}.
\end{align}
After introducing the Bures angle (\ref{eq: define Bures angle}),
the above equation reads
\begin{align}
&
\left|
\frac{\pi}{2}
\frac{\partial}{\partial\lambda}
D(\Psi^{\,}_{\lambda},\Psi^{\,}_{0})
\right|
\leq
\frac{\Delta E^{\,}_{0}}{\hbar|\partial^{\,}_{t}\lambda|}
\qquad
\Rightarrow
\qquad
-\frac{\Delta E^{\,}_{0}}{\hbar|\partial^{\,}_{t}\lambda|}
\leq
\frac{\pi}{2}
\frac{\partial}{\partial\lambda}
D(\Psi^{\,}_{\lambda},\Psi^{\,}_{0})
\leq
\frac{\Delta E^{\,}_{0}}{\hbar|\partial^{\,}_{t}\lambda|},
\end{align}
which can be integrated over $\lambda$ to obtain (recall $D(\Psi^{\,}_{\lambda},\Psi^{\,}_{0})\geq0$ by definition)
\begin{align}
&
\frac{\pi}{2}D(\Psi^{\,}_{\lambda},\Psi^{\,}_{0})
\leq
\frac{1}{\hbar}
\int^{\lambda}_{0}
\frac{\mathrm{d}\lambda'}{|\partial^{\,}_{t}\lambda'|}
\Delta E^{\,}_{0}(\lambda').
\label{eq: QSL proof type a}
\end{align}
Hence, the proof of Eq.\ (\ref{eq: bound from QSL}) is complete.

Note that in Eq.\ (\ref{eq: manipulate the rate of change}), if we instead apply the key formula (\ref{eq: a key formula}) 
to $\langle\Psi^{\,}_{\lambda}|H^{\,}_{\lambda}$, 
we shall obtain the following version of the quantum speed limit
\begin{align}
&
\frac{\pi}{2}D(\Psi^{\,}_{\lambda},\Psi^{\,}_{0})
\leq
\frac{1}{\hbar}
\int^{\lambda}_{0}
\frac{\mathrm{d}\lambda'}{|\partial^{\,}_{t}\lambda'|}
\Delta E(\lambda'),
\label{eq: QSL proof type b}
\end{align}
where $\Delta E(\lambda)=\sqrt{\langle\Psi^{\,}_{\lambda}|H^{2}_{\lambda}|\Psi^{\,}_{\lambda}\rangle-\langle\Psi^{\,}_{\lambda}|H^{\,}_{\lambda}|\Psi^{\,}_{\lambda}\rangle^2}$
is the quantum uncertainty of $H^{\,}_{\lambda}$ with respect to $|\Psi^{\,}_{\lambda}\rangle.$ 
For a time-independent Hamiltonian, Eq.\ (\ref{eq: QSL proof type b}) is identical to Eq.\ (\ref{eq: QSL proof type a}). 
This type of quantum speed limit is known as {\it Mandelstam-Tamm inequality} in literature \cite{Mandelstam45,Deffner17}.

\section{Direct proof of the inequality (\ref{eq: new inequality})}
\label{sec: direct proof of the inequality}

In this section, we want to prove the inequality (\ref{eq: new inequality}) using an elementary method.
In terms of the fidelity (\ref{eq: define fidelity}) and the Bures angle (\ref{eq: define Bures angle}),
the inequality (\ref{eq: new inequality}) reads
\begin{align}
\left|
F\left(\Phi^{\,}_{\lambda},\Psi^{\,}_{\lambda}\right)-F\left(\Phi^{\,}_{\lambda},\Phi^{\,}_{0}\right)
\right|
\quad\leq\quad
\sin\left(
\frac{\pi}{2}
D(\Psi^{\,}_{\lambda},\Psi^{\,}_{0})
\right).
\end{align}
Indeed, this inequality is an instance of the following lemma \cite{Rastegin02a,Rastegin02b}:
\begin{lemma}
For each triplet $\{|\psi^{\,}_{1}\rangle,|\psi^{\,}_{2}\rangle,|\psi^{\,}_{3}\rangle\}$ of states, the following inequality holds:
\begin{align}
\left|
F\left(\psi^{\,}_{1},\psi^{\,}_{2}\right)-F\left(\psi^{\,}_{1},\psi^{\,}_{3}\right)
\right|
\leq 
\sin\left(
\frac{\pi}{2}
D(\psi^{\,}_{2},\psi^{\,}_{3})
\right),
\label{eq: an important inequality}
\end{align}
where $F(.,.)$ is the fidelity (\ref{eq: define fidelity}) and $D(.,.)$ is the Bures angle (\ref{eq: define Bures angle}).
\end{lemma}
\begin{proof}
Apply the following triangle inequality
\begin{align}
D(\psi^{\,}_{2},\psi^{\,}_{3})\leq D(\psi^{\,}_{1},\psi^{\,}_{2}) + D(\psi^{\,}_{1},\psi^{\,}_{3}),
\label{eq: triangle inequality prime}
\end{align}
and the trigonometric formula, $\cos^2x-\cos^2y=\sin(x+y)\sin(y-x)$, 
to the left side of Eq.\ (\ref{eq: an important inequality})
\begin{align}
F\left(\psi^{\,}_{1},\psi^{\,}_{2}\right)-F\left(\psi^{\,}_{1},\psi^{\,}_{3}\right)
&=
\cos^2
\left(
\frac{\pi}{2}
D(\psi^{\,}_{1},\psi^{\,}_{2})
\right)
-
\cos^2
\left(
\frac{\pi}{2}
D(\psi^{\,}_{1},\psi^{\,}_{3})
\right)
\nonumber
\\
&\leq
\cos^2
\left(
\frac{\pi}{2}
D(\psi^{\,}_{2},\psi^{\,}_{3})
-
\frac{\pi}{2}
D(\psi^{\,}_{1},\psi^{\,}_{3})
\right)
-
\cos^2
\left(
\frac{\pi}{2}
D(\psi^{\,}_{1},\psi^{\,}_{3})
\right)
\nonumber
\\
&=
\sin
\left(
\frac{\pi}{2}
D(\psi^{\,}_{2},\psi^{\,}_{3})
\right)
\sin
\left(
\pi
D(\psi^{\,}_{1},\psi^{\,}_{3})
-
\frac{\pi}{2}
D(\psi^{\,}_{2},\psi^{\,}_{3})
\right)
\nonumber
\\
&\leq
\sin
\left(
\frac{\pi}{2}
D(\psi^{\,}_{2},\psi^{\,}_{3})
\right).
\label{eq: prove an inequality a}
\end{align}
Exchange the role of $|\psi^{\,}_{2}\rangle$ and $|\psi^{\,}_{3}\rangle$ in Eq.\ (\ref{eq: prove an inequality a}) yields
\begin{align}
F\left(\psi^{\,}_{1},\psi^{\,}_{3}\right)-F\left(\psi^{\,}_{1},\psi^{\,}_{2}\right)
&=
\cos^2
\left(
\frac{\pi}{2}
D(\psi^{\,}_{1},\psi^{\,}_{3})
\right)
-
\cos^2
\left(
\frac{\pi}{2}
D(\psi^{\,}_{1},\psi^{\,}_{2})
\right)
\nonumber
\\
&\leq
\sin
\left(
\frac{\pi}{2}
D(\psi^{\,}_{2},\psi^{\,}_{3})
\right).
\label{eq: prove an inequality b}
\end{align}
Upon combing Eqs.\ (\ref{eq: prove an inequality a}) and (\ref{eq: prove an inequality b}),
the proof of Eq.\ (\ref{eq: an important inequality}) is therefore complete.
\end{proof}

\end{document}